\documentclass[runningheads]{llncs}
\usepackage[utf8]{inputenc}
\usepackage{lineno}
\linenumbers

\usepackage[UKenglish]{babel}
\usepackage{mathtools}
\usepackage{verbatim}
\usepackage{semantic}
\usepackage{array}
\usepackage{tikz}
\usepackage[breaklinks=true]{hyperref}
\usepackage{breakcites}
\usetikzlibrary{calc,trees,positioning,arrows,chains,shapes.geometric,%
    decorations.pathreplacing,decorations.pathmorphing,shapes,%
    matrix,shapes.symbols}

\begin{document}

\title{Data Type Inference for Logic Programming}
\author{João Barbosa \and Mário Florido \and Vítor Santos Costa}
\authorrunning{J. Barbosa, M. Florido, V. Santos Costa}
\institute{Faculdade de Ciências, Universidade do Porto, Portugal \\
\email{\{joao.barbosa,amflorid,vscosta\}@fc.up.pt}}

\maketitle

\begin{abstract}
In this paper we present a new static data type inference algorithm for logic programming. Without the need of declaring types for predicates, our algorithm is able to automatically assign types to predicates which, in most cases, correspond to the data types processed by their intended meaning. The algorithm is also able to infer types given data type definitions similar to data definitions in Haskell and, in this case, the inferred types are more informative in general.
We present the type inference algorithm, prove some properties and finally, we evaluate our approach on example programs that deal with different data structures.
\end{abstract}


\section{Introduction}
Types are program annotations that provide information about data
usage and program execution.  Ensuring that all types are
correct and consistent may be a daunting task for humans. However,
this task can be automatized with the use of a type inference algorithm which  assigns types to programs.

Logic programming implementers have been interested in types
from early on~\cite{DBLP:conf/iclp/Zobel87,DBLP:books/mit/pfenning92/DartZ92,DBLP:journals/corr/cs-LO-9810001,DBLP:conf/lics/FruhwirthSVY91,DBLP:books/mit/pfenning92/YardeniFS92,DBLP:journals/ai/MycroftO84,DBLP:conf/slp/LakshmanR91,DBLP:conf/lopstr/SchrijversBG08,DBLP:books/mit/pfenning92/HeintzeJ92,DBLP:conf/iclp/SchrijversCWD08}. Most research approached typing an over-approximation (a superset) of the program semantics \cite{DBLP:conf/iclp/Zobel87,DBLP:books/mit/pfenning92/DartZ92,DBLP:books/mit/pfenning92/YardeniFS92,DBLP:conf/iclp/BruynoogheJ88,DBLP:conf/lics/FruhwirthSVY91}: any programs that succeeds
will necessarily be well-typed. Other researchers followed the
experience of functional languages and took a more aggressive approach
to typing, where only well-typed programs are acceptable. 
Over the course of the last years it has become clear that there is a
need for a type inference system that can support Prolog well. Next, we report on recent progress on our design, the
$YAP^T$ type system\footnote{This work is partially funded by the portuguese Fundação para a Ciência e a Tecnologia and by LIACC (FCT/UID/CEC/0027/2020).}. We will introduce the
key ideas and then focus on the practical aspects.

 Our approach is motivated by the
belief that programs (Prolog or otherwise) are about manipulating data
structures. In Prolog, data structures are denoted by terms with a common structure, and, being untyped 
one cannot naturally distinguish between failure and results of type erroneous calls. We believe that to fully use data structures we must be able to discriminate between failure, error, and success~\cite{BarbosaFloridoCosta19}.

Our starting point was a three-valued semantics that clearly distinguishes error and falsehood. In \cite{BarbosaFloridoCosta19} we use it do define a type system for logic programming that is
semantically sound. In \cite{BarbosaFloridoCosta19} we presented a type system which formally defines the meaning of a {\em well-typed} program and showed that this notion of {\em well-typing} is sound with respect to a three-valued program semantics. From now on when we say {\em well-typed} we mean the notion presented in \cite{BarbosaFloridoCosta19}. Here we present the $YAP^T$ type inference algorithm which is able to automatically infer data type definitions.

To use our semantics, we shall assume that typed Prolog programs
operate in a context, e.g., suppose a programming context where the well-known
\texttt{append} predicate is expected to operate on lists:
\begin{verbatim}
append([],X,X).
append([X|R],Y,[X|R1]) :- append(R,Y,R1).
\end{verbatim}
This information is not achievable when using type inference as a
conservative approximation of the success set of the predicate.
The following example shows the output of type inference in this case, where \texttt{ti} is the type of the {\em i-th} argument of \texttt{append}, ``+" means type disjunction and ``A" and ``B" are type variables:
\begin{figure}
\centering
\begin{minipage}{.5\textwidth}
  \begin{verbatim}
t1 = [] + [A | t1]
t2 = B
t3 = B + [A | t3]
  \end{verbatim}
\end{minipage}%
\begin{minipage}{.5\textwidth}
  \begin{verbatim}
t1 = [] + [A | t1]
t2 = [] + [A | t2]
t3 = [] + [A | t3]
  \end{verbatim}
\end{minipage}
\caption{(1) Program approximation ; (2) Well-typing}
\end{figure}

Types {\ttfamily t2} and {\ttfamily t3}, for the second and third argument of the left-hand side (1), do not filter any possible term, since they have a type variable as a member of the type definition, which can be instantiated with any type. And, in fact, assuming the specific context of using \texttt{append} as list concatenation, some calls to \texttt{append} succeed even if {\em unintended}\footnote{accordingly to a notion of {\em intended meaning} first presented in \cite{Naish92}}, such as \texttt{append([],1,1)}. The solution we found for these arguably over-general types is the definition of closed types, that we first presented in \cite{BarbosaFloridoCosta17}, which are types where every occurrence of a type variable is constrained. We also defined a closure operation, from open types into closed types, using only information provided by the set of types themselves.
Applying our type inference algorithm with closure to the \texttt{append} predicate yields the types on the right-hand side (2), which are the {\em intended types} \cite{Naish92} for the \texttt{append} predicate.

Our type inference algorithm\footnote{Implementation at \url{https://github.com/JoaoLBarbosa/TypeInferenceAlgorithm}} works for pure Prolog with system predicates for arithmetic. We assume as base types $int$, $float$, $string$, and $atom$. There is an optional definition of type declarations (like data declarations in Haskell), which, if declared by the programmer, are used by the type inference algorithm to refine types. We follow a syntax inspired in \cite{DBLP:conf/iclp/SchrijversCWD08} to specify type information. One example of such a declaration is the list datatype \begin{center}\texttt{:- type list(A)~=~[]~+~[A~|~list(A)]}.\end{center}

In order to simplify further processing our type inference algorithm first compiles every predicate to a simplified form called {\em kernel Prolog} \cite{VanRoy:CSD-90-600}. In this representation, each predicate is defined by a single clause ($H :- B$), where the head $H$ contains distinct variables as arguments and the body $B$ is a disjunction (represented by the symbol $;$) of queries. We assume that there are no common variables between queries, except for the variables that occur in the head of the clause, without loss of generality.  In this form the scope of variables is not limited to a single clause, but is extended over the whole predicate and thus type inference is easier to perform. In the rest of the paper we will assume that predicate definitions are always in kernel Prolog.

As simple illustrating examples, using a pre-defined type declaration for lists, and reading "::" as "has type", our type inference algorithm gives the following results:
\begin{verbatim}
PREDICATE:                             TYPE:

l(X) :- (X = [] ; X = [Y|Ys], l(Ys)).   l :: t1
                                        t1 = list(C)
                                        list(B) = [] + [B|list(B)]

p(X) :- X = a ; X = 3.                  p :: t2
                                        t2 = atom + int

q(Y) :- Y =  1.23; Y = 5.               q :: t3
                                        t3 = float + int
                                        
h(Z) :- Z = W, Z = Q, p(W), q(Q).       h :: t4
                                        t4 = int
\end{verbatim}

\section{Types}\label{regtype}

Here we define a new class of expressions, which we shall call {\em types}. We first define the notion of {\em type term} built from an infinite and enumerable set of type variables  $TVar$, a finite set of base types $TBase$, an infinite and enumerable set of constants $TCons$, an infinite and enumerable set of function symbols $TFunc$, and an infinite and enumerable set of type symbols, $TSymb$. 
{\em Type terms} can be:
\begin{itemize}
\item a type variable ($\alpha,\beta,\gamma,\dots \in TVar$)
\item a type constant ($1, [~], `c\text{'},\dots \in TCons$)
\item a base type ($int, float, \dots \in TBase$)
\item a function symbol $f \in TFunc$ associated with an arity $n$ applied to an n-tuple of type terms ($f(int,[~], g(X))$).
\item a type symbol $\sigma \in TSymb$ associated with an arity $n$ ($n \ge 0)$ applied to an n-tuple of type variables ($\sigma(X,Y)$).
\end{itemize}

Type variables, constants and base types are called {\em basic types}. A {\em ground type term} is a variable-free type term. 
A {\em type definition}, which introduces a new type called an {\em algebraic data type}, is of the form:

$$\sigma(X_1,\ldots,X_k) = \tau_1 + \dots + \tau_n,$$
where each $\tau_i$ is a type term and $\sigma$ is the type symbol being defined. Variables $X_1,\ldots,X_n$, for $n \geq 0$, include the type variables occurring in $\tau_1 + \dots + \tau_n$, and are called {\em type parameters}. The sum $\tau_1 + \dots + \tau_n$ is a {\em union type}, describing values that may have one of the types $\tau_1,\ldots,\tau_n$. Note that type definitions may be recursive. {\em Deterministic type definitions} are type definitions where, on the right-hand side, none of $\tau_i$ is a type symbol and if $\tau_i$ is a type term starting with a type function symbol $f$, then no other $\tau_j$ starts with $f$.

\begin{example}
Assuming a base type $int$ for the set of all integers, the type list of integers is defined by the type definition 
$list = [~] + [int~|~list]$\footnote{Type definitions will use the user friendly Prolog notation for lists instead of the list constructor}.
\end{example}
A {\em predicate type} is a functional type from a tuple of the type terms defining the types of its arguments to bool: $\tau_1 \times \ldots \times \tau_n \to bool$. A {\em type} can be a {\em type term}, an {\em algebraic data type} or a {\em predicate type}. In \cite{BarbosaFloridoCosta19} we defined a formal Hindley-Milner  semantics for types.

Our type language enables parametric polymorphism through the use of type schemes. A {\em type scheme} is defined as $\forall_{X_1}\ldots\forall_{X_n} T$, where $T$ is a predicate type and $X_1,\ldots,X_n$ are distinct type variables.
In logic programming, there have been several authors that have dealt with polymorphism with type schemes or in a similar way \cite{DBLP:conf/slp/PyoR89,DBLP:journals/scp/BarbutiG92,Henglein:1993:TIP:169701.169692,DBLP:conf/iclp/Zobel87,DBLP:conf/lics/FruhwirthSVY91,DBLP:conf/iclp/GallagherW94,DBLP:books/mit/pfenning92/YardeniFS92,TAT}. Type schemes have type variables as generic place-holders for ground type terms. Parametric polymorphism comes form the fact these type variables can be instantiated with any type. 

\begin{example}
A polymorphic list is defined by the following type definition:

$$list(X) = [~] + [X~|~list(X)]$$
\end{example}

\noindent {\bf Notation} Throughout the rest of the paper, for the sake of readability, we will omit the universal quantifiers on type schemes and the type parameters as explicit arguments of type symbols in inferred types. Thus we will assume that all free type variables on type definitions of inferred types are type parameters which are universally quantified.


Most  type  languages  in  logic  programming  use  tuple  distributive  closures of types. The notion of tuple distributivity was given by Mishra \cite{Mishra84}. Throughout this paper, we restrict our type language to tuple distributive types, where type definitions are deterministic.  

Sometimes, the programmer wants to introduce a new type in a program, so that it is recognized when performing type inference. It is also a way of having a more structured and clear program. These declarations act similarly to data declarations in Haskell.

In our algorithm, types can be declared by the programmer in the following way $:- type ~ type\_symbol(type\_vars) = type\_term_1 + \dots + type\_term_n$. One example would be:
\begin{verbatim}
:- type tree(X) = empty + node(X, tree(X), tree(X)).
\end{verbatim}

In the rest of the paper we will assume that all constants and function symbols that start a summand in a declared type cannot start a summand in a different one, thus there are no overloaded constants nor function symbols. Note that there is a similar restriction on data declarations in functional programming languages.

\subsection{Closed Types}

{\em Closed types} were first defined in \cite{BarbosaFloridoCosta17}. Informally, they are types where every occurrence of a type variable is constrained. If a type is not closed, we say that it is an {\em open type}.
The restrictions under the definition of closed type can be compressed in the following three principles:
\begin{itemize}
    \item Types should denote a set of terms which is strictly smaller than the set of all terms
    \item Every use of a variable in a program should be type constrained
    \item Types are based on self-contained definitions
\end{itemize}
The last one is important to create a way to go from open types to closed types.
We defined what is an unconstrained type variable as follows:
\begin{definition}[Unconstrained Type Variable]
A type variable $\alpha$ is {\em unconstrained} with respect to a set of type definitions $T$, notation $unconstrained(\alpha,T)$, if and only if it occurs exactly once in the set of all the right-hand sides of type definitions in $T$.
\end{definition}
Unconstrained type variables type terms with any type, thus they do not really provide type information. 
We now define {\em closed algebraic data types}, which are types 
without variables as summands in their definition.

\begin{definition}[Closed Algebraic Data Type]
An algebraic data type $\tau$ is {\em closed}, notation \emph{closedDataType}$(\tau)$, if and only if it has no type variables as summands in its type definition.
\end{definition}

The definition for closed types uses these two previous auxiliary definitions. Closed types correspond to close records or data definitions in functional programming languages. The definition follows:

\begin{definition}[Closed Types]
A type term $\tau$ is {\em closed} with respect to a set of type definitions $T$, notation \emph{closed}$(\tau,T)$, if and only if the predicate defined as follows holds:
\begin{equation*}
    closed(\tau,T) =
    \begin{dcases*}
      closedDataType(\tau)  & if $\tau$ is an algebraic data type \\
      \neg unconstrained(\tau,T) & if $\tau$ is a type variable   \\
      True & if $\tau$ is basic but not a type variable  \\
    \end{dcases*}
 \end{equation*}

\end{definition}

\begin{example}
We recall the example in the Introduction, of the following types for the \texttt{append} predicate, where $t_n$ is the type of $nth$ predicate argument:
\begin{verbatim}
t1 = [] + [A | t1]
t2 = B
t3 = B + [A | t3]
\end{verbatim}
Type \texttt{t3}, for the third argument of \texttt{append}, is open, because \texttt{t3} has a type variable as a summand, thus it does not filter any possible term, since the type variable can be instantiated with any type. 
An example of a valid closed type for \texttt{append} is:
\begin{verbatim}
t1 = [] + [A | t1]
t2 = [] + [A | t2]
t3 = [] + [A | t3]
\end{verbatim}

\end{example}

The next step is to transform open types into closed types. Note that some inferred types may be already closed. For the ones that are not, we defined a {\em closure} operation, described in detail in \cite{BarbosaFloridoCosta17}. This closure operation is present in our type inference algorithm as an optional step to perform on the resulting types after solving the type constraints generated by the algorithm.

\section{Examples} \label{examples}

There are some flags in the type inference algorithm that can be turned on or off:
\begin{itemize}
    \item {\em basetype} (default: {\em on}) - this flag types constants with their base types when turned on and with a constant symbol equal to the constant itself otherwise;
    \item {\em list} (default: {\em off}) - this flag adds the data type declaration for polymorphic lists to the program when turned on;
    \item {\em closure} (default: {\em off}) - this flag applies closure to the resulting type definitions after solving the constraints when turned on.
\end{itemize}
 In the following examples \texttt{pi} is the type symbol for the type of the $ith$ argument of predicate \texttt{p} and we assume that all free type variables on type definitions are universally quantified and that the type of arguments of built-in arithmetic predicates is predefined as $int + float$.

\begin{example}
Let us consider the predicate $concat$, which flattens a list of lists, where {\em app} is the {\em append} predicate:
\begin{verbatim}
concat(X1,X2) :- X1=[], X2=[];
        X1=[X|Xs], X2=List, concat(Xs,NXs), app(X,NXs,List).

app(A,B,C) :- A=[], B=D, C=D;
        app(E,F,G), E=H, F=I, G=J, A=[K|H], B=I, C=[K|J].
\end{verbatim}
The types inferred with all the flags {\em off} correspond to types inferred in previous type inference algorithms which view types as an approximation of the success set of the program:
\begin{verbatim}
concat :: concat1 x concat2
concat1 = [] + [ t | concat1 ]
concat2 = C + [] + [ B | concat2 ]
t = [] + [ B | t ]

app :: app1 x app2 x app3
app1 = [] + [ A | app1 ]
app2 = B
app3 = B + [ A | app3 ]
\end{verbatim}
Now the types inferred when turning on the closure flag, \emph{closure\_flag}, are:
\begin{verbatim}
concat :: concat1 x concat2
concat1 = [] + [ concat2 | concat1 ]
concat2 = [] + [ B | concat2 ]

app :: app1 x app2 x app3
app1 = [] + [ A | app1 ]
app2 = [] + [ A | app2 ]
app3 = [] + [ A | app3 ]
\end{verbatim}
Note that these types are not inferred by any previous type inference algorithm for logic programming so far, and they are a step forward to the automatic inference of program use in a specific context, more precisely, a context which corresponds to how it would be used in a programming language with data type declarations, such as Curry \cite{DBLP:conf/birthday/Hanus13} or Haskell.
\end{example}

\begin{example}
Let $rev$ be the reverse list predicate, defined using the $append$ definition used in the previous example:
\begin{verbatim}
rev(A, B) :- A=[], B=[] ;
        rev(C, D), app(D, E, F), E=[G], A=[G|C], B=F.
\end{verbatim}
The inferred types with all flags off is (the types inferred for append are the same as the one in the previous example):
\begin{verbatim}
rev :: rev1 x rev2
rev1 = [] + [ A | rev1 ]
rev2 = [] + [ t | rev2 ]
t = B + A
\end{verbatim}
If we turn on the \emph{list\_flag}, which declares the data type for Prolog lists, the type inference algorithm outputs the same types that would be inferred in Curry or Haskell with pre-defined built-in lists:
\begin{verbatim}
rev :: rev1 x rev2
rev1 = list(A)
rev2 = list(A)

list(X) = [] + [ X | list(X) ]
\end{verbatim}
\end{example}
We now show an example of the minimum of a tree.
\begin{example}
Let $tree\_minimum$ be the predicate defined as follows:
\begin{verbatim}
tree_min(A,B) :- A=empty, B=0 ;
        A=node(C,D,E), tree_min(D,F), tree_min(E,G),
        Y=[C,F,G], minimum(Y,X), X=B.

minimum(A,B) :- A=[I], B=I;
        A=[X|Xs], minimum(Xs,C), X=<C, B=C ;
        A=[Y|Ys], minimum(Ys,D), D=<Y, B=D.
\end{verbatim}
The inferred types with all flags off, except for the $basetype\_flag$, are:
\begin{verbatim}
tree_min :: tree_min1 x tree_min2
tree_min1 = atom + node(tree_min2, tree_min1, tree_min1)
tree_min2 = A + int + float

minimum :: minimum1 x minimum2
minimum1 = [ minimum2 | t ]
minimum2 = A + int + float
t2 = [] + [ minimum2 | t2 ]
\end{verbatim}
If we now add a predefined declaration of a tree data type and turn on the $list\_flag$, the algorithm outputs:
\begin{verbatim}
tree_minimum :: tree_minimum1 x tree_minimum2
tree_minimum1 = tree(tree_minimum2)
tree_minimum2 = int + float

minimum :: minimum1 x minimum2
minimum1 = list(minimum2)
minimum2 = int + float

tree(X) = empty + node(X, tree(X), tree(X))
list(Y) = [] + [ Y | list(Y) ]
\end{verbatim}
\end{example}

\section{Type Inference} \label{algorithm}

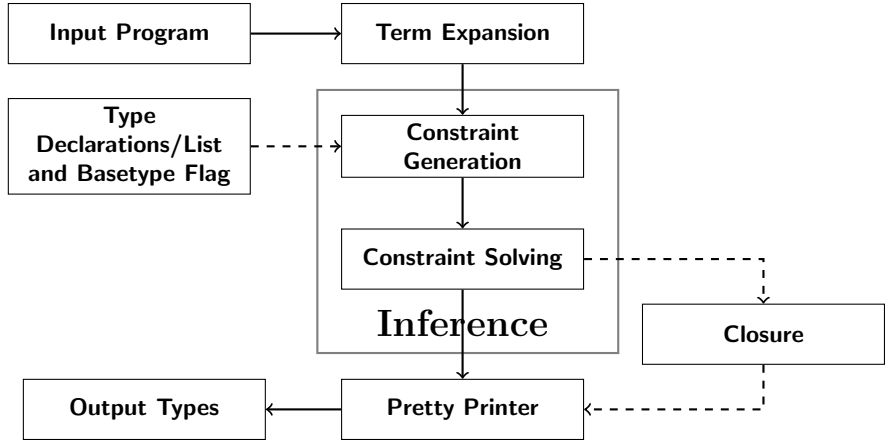
\begin{figure}[ht] 
\centering
\begin{tikzpicture}
 [node distance = 1cm, auto,font=\footnotesize,
every node/.style={node distance=1.5cm},
node/.style = {rectangle,draw,text width=3cm,text badly centered, minimum height=0.8cm,font=\bfseries\footnotesize\sffamily}]
 
 \node [node] (input) {Input Program};
 \node [node, right= 1.2cm of input] (expand) {Term Expansion};
 \node [node, below of = expand] (congen) {Constraint Generation};
 \node [node, left= 1.2cm of congen] (flags) {Type Declarations/List and Basetype Flag};
 \node [node, below of = congen] (solver) {Constraint Solving};
 \draw (solver)+(4,-1) node[rectangle,draw,text width=3cm,text badly centered, minimum height=0.8cm,font=\bfseries\footnotesize\sffamily] (closure) {Closure};
 \node [node, node distance = 2cm, below of=solver] (pp) {Pretty Printer};
 \node [node, left= 1cm of pp] (output) {Output Types};
 \path (pp)+(0,1.15) node[font=\bfseries\textbf{}] (algorithm) {\Large{Inference}};
 
 \draw[color=gray,thick](2.5,-4.25) rectangle (6.5,-.75);
 
\draw[->, thick, dashed]
(solver.east) -| (closure.north);
 
\draw[->, thick, dashed]
(closure.south) |- (pp.east);

\draw[->, thick, dashed]
(flags) to (congen);

\path[->,thick]
(input) edge (expand)
(expand) edge (congen)
(congen) edge (solver)
(solver) edge (pp)
(pp) edge (output);
\end{tikzpicture}
\caption{Type Inference Algorithm Flowchart}
\label{imagem}
\end{figure}

The type inference algorithm is composed of several modules, as described in Figure \ref{imagem}. On a first step, when consulting programs, we apply term expansion to transform programs into the internal format that the rest of the algorithm expects. Secondly, we have the type inference phase itself, where a type constraint solver outputs the inferred types for a given program. There is also a simplification phase performed during inference. After this, we run directly a type pretty printer, or go through closure before printing the types.

Thus the type inference algorithm is composed of four main parts with some auxiliary steps:
\begin{itemize}
    \item Term expansion
    \item Constraint generation
    \item Constraint solving
    \item Closure (optional)
\end{itemize}

There is also a simplification step, that is applied at every step.
Without closure or type declarations our algorithm follows a standard approach of types as approximations of the program semantics. Using our algorithm to infer well-typings (which filter program behaviour instead of approximating it) is possible either by using explicit type declarations or by using the closure step.

\subsection{Stratification}

We assume that the input program of our algorithm is {\em stratified}. To understand the meaning of stratified programs, let us define the {\em dependency directed graph} of a program as the graph that has one node representing each predicate in the program and an edge from $q$ to $p$ for each call from a predicate $p$ to a predicate $q$.

\begin{definition}[Stratified Program]
A {\em stratified program} $P$ is such that the dependency directed graph of $P$ has no cycles of size more than one.
\end{definition}

This means that our type inference algorithm deals with predicates defined by direct recursion but not with mutual recursion. Note that stratified programs are widely used and characterize a large class of programs which is used in several database and knowledge base systems \cite{Ullman88}. Despite this, expanding our algorithm to mutually recursive predicates is a step we are working towards.

\subsection{Constraints and Constraint Generation}

The type inference algorithm begins by generating type constraints from a logic program, that will be solved by a constraint solver in a second stage of the algorithm.
There are two different kinds of type constraints: equality constraints and subtyping constraints. The constraint generation step of the algorithm will output two sets of constraints, $Eq$ (a set of equality constraints) and $Ineq$ (a set of subtyping constraints), that need to be solved during type inference.
An equality constraint is of the form $\tau_1 = \tau_2$ and a subtyping constraint is of the form $\tau_1 \leq \tau_2$. In the subtyping constraint case, $\tau_i$ can be a type term or a sum of several type terms.

An {\em assumption} is a type declaration for a variable, written $X:\tau$, where $X$ is a variable and $\tau$ a type. We define a {\em context} $\Gamma$ as a set of assumptions with distinct variables as subjects (alternatively {\em contexts} can be defined as functions from variables to types and where $domain(\Gamma)$ stands for its domain).
Contexts are synthetized by the type inference algorithm in two different ways: we can build a new context by linking assumptions in different contexts with disjunction (corresponding to different clauses of the same predicate) or with conjunction (corresponding to different queries). For this we define two auxiliary functions, $\otimes$ (for conjunction) and $\oplus$ (for disjunction), which are used by the constraint generation algorithm. They are defined as follows:

\begin{definition}
Let $\Gamma_1$ and $\Gamma_2$ be two contexts and $Def$ be a set of type definitions defining the type symbols in $\Gamma_1 \cup \Gamma_2$. Let $V = domain(\Gamma_1) \cap domain(\Gamma_2)$.\\
$\oplus \big{(}(\Gamma_1 \cup \Gamma_2), Def\big{)} = (\Gamma,Def\prime)$, such that $Def \subseteq Def\prime$, \\
where $\forall X \in V.\Gamma(X) = \sigma$ and $(\sigma = \Gamma_1(X) + \Gamma_2(X)) \in Def\prime$, all $\sigma$s fresh, and \\
$\forall X. X \in domain(\Gamma_i) , X \notin V. \Gamma(X) = \Gamma_i(X)$, for $i,j = 1,2$.
\end{definition}

\begin{definition}
Let $\Gamma_1$ and $\Gamma_2$ be two contexts and $Def$ be a set of type definitions defining the type symbols in $\Gamma_1 \cup \Gamma_2$. Let $V = domain(\Gamma_1) \cap domain(\Gamma_2)$.\\
$\otimes \big{(} (\Gamma_1 \cup \Gamma_2), Def\big{)} = (\Gamma,Def\prime,Eq)$\\
where $\forall X \in V.\Gamma(X) = \sigma$, $(\sigma = \alpha) \in Def\prime$, and $\{\alpha = def(\Gamma_1(X)), \alpha = def(\Gamma_2(X))\} \in Eq$, all $\alpha$s and $\sigma$s fresh, and \\
$\forall X. X \in domain(\Gamma_i) , X \notin V. \Gamma(X) = \Gamma_i(X)$, for $i,j = 1,2$.
\end{definition}

Let $P$ be a term, an atom, a query, a sequence of queries or a clause. $generate(P)$ is a function that outputs a tuple of the form $(\tau,\Gamma,Eq,Ineq,Def)$, where $\tau$ is a type, $\Gamma$ is an context for logic variables, $Eq$ is a set of equality constraints, $Ineq$ is a set of subset constraints, and $Def$ is a set of type defintions. The function $generate$, which generates the initial type constraints, is defined by cases from the program syntax. Its definition follows:
\newline
\newline
$generate(P) =$
\begin{itemize}
    \item $generate(X) = (\alpha,\{X:\sigma\},\emptyset,\emptyset,\{\sigma = \alpha\})$, $X$ is a variable,\\
    where $\alpha$ is a fresh type variable and $\sigma$ is a fresh type symbol.\\
    
    \item $generate(c) = (basetype(c),\emptyset,\emptyset,\emptyset,\emptyset)$, $c$ is a constant.\\
    
    \item $generate(f(t_1,\dots,t_n)) = (basetype(f)(\tau_1,\dots,\tau_n), \Gamma, Eq, \emptyset, Def)$,\\
    where $generate(t_i) = (\tau_i,\Gamma_i,Eq_i,\emptyset,Def_i)$,\\
    $(\Gamma,Eq\prime,Def) = \otimes \big{(}(\Gamma_1 \cup \dots \cup \Gamma_n), (Def_1 \cup \dots Def_n)\big{)}$, and\\
    $Eq = Eq_1 \cup \dots \cup Eq_n \cup Eq\prime$.\\
    
    \item $generate(t_1 = t_2) = (bool, \Gamma, Eq, \emptyset,Def)$\\
    where $generate(ti) = (\tau_i,\Gamma_i, Eq_i, \emptyset,Def_i)$\\
    $(\Gamma,Def,Eq\prime) = \otimes \big{(}(\Gamma_1 \cup \Gamma_2), (Def_1 \cup Def_2)\big{)}$, and\\
    $Eq = Eq_1 \cup Eq_2 \cup \{\tau_1 = \tau_2\} \cup Eq\prime$.\\
    
    \item $generate(p(t_1,\dots,t_n)) = (bool, (\{X_1:\sigma_1, \dots, X_n : \sigma_n\}, Eq, \{\alpha_1 \leq \tau_1, \dots, \alpha_n \leq \tau_n\},Def\prime)$, a non-recursive call, \\ where $generate(p(Y_1,\dots,Y_n) :- body) = (bool, \Gamma, Eq, Ineq,Def)$,\\ $\{Y_1 : \tau_1, \dots Y_n : \tau_n\} \in \Gamma$\\
    $Def\prime = Def \cup \{\sigma_i = \alpha_i\}$, and $\sigma_i$ and $\alpha_i$ are all fresh.\\
    
    \item $generate(p(t_1,\dots,t_n)) = (bool, \{X_1 : \alpha_1,\dots,X_n:\alpha_n\}, \emptyset,\{\alpha_1 \leq \tau_1, \dots, \alpha_n \leq \tau_n \}, \emptyset)$, a recursive call,\\
    where $\tau_1, \dots, \tau_n$ are the types for the variables in the head of the clause this call occurs in, and $\alpha_i$ are all fresh type variables.\\
    
    \item $generate(c_1,\dots,c_n) = (bool, \Gamma, Eq , Ineq_1 \cup \dots \cup Ineq_n,Def)$,\\
    where $generate(c_i) = (bool,\Gamma_i,Eq_i,Ineq_i,Def_i)$,\\
    $(\Gamma,Def,Eq\prime) = \otimes \big{(} (\Gamma_1 \cup \dots \cup \Gamma_n), (Def_1 \cup \dots \cup Def_n)\big{)}$, and\\
    $Eq = Eq_1 \cup \dots \cup Eq_n \cup Eq\prime$.\\
    
    \item $generate(b_1;\dots;b_n) = (bool,\Gamma, Eq_1 \cup \dots \cup Eq_n, Ineq_1 \cup \dots \cup Ineq_n,Def)$,\\
    where $generate(c_i) = (bool,\Gamma_i,Eq_i,Ineq_i,Def_i)$, and\\
    $(\Gamma,Def) = \oplus \big{(} (\Gamma_1 \cup \dots \cup \Gamma_n), (Def_1 \cup \dots \cup Def_n)\big{)}$.\\
    
    \item $generate(p(X_1,\dots,x_n) :- body.) = (bool, \Gamma, Eq, Ineq, Def)$,\\
    where $generate(body) = (bool, \Gamma, Eq, Ineq, Def)$.
\end{itemize}

\begin{example} \label{mainexample}
Consider the following predicate:
\begin{verbatim}
list(X) :- X = []; X = [Y|YS], list(Ys).
\end{verbatim}
the output of applying the generate function to the predicate is:\\
$generate(list(X) :- X = [~]; X = [Y|YS], list(Ys)) = \\ 
\{bool, \{X:\sigma_1, Y:\sigma_2, Ys:\sigma_3\}, \{\alpha = [~], \beta = [\delta~|~\epsilon], \epsilon  = \upsilon\}, \{\upsilon \leq \sigma_1\}, \{\sigma_1 = \alpha + \beta, \sigma_2 = \delta, \sigma_3 = \epsilon, \sigma_4 = \upsilon \}\}$\\
The set $\{\upsilon \leq \sigma_1\}$ comes from the recursive call to the predicate, while $\alpha = [~]$ comes from $X = [~]$, and $\beta = [\delta~|~\epsilon]$ comes from $X = [Y|YS]$. The other constraints come from the application of the $\otimes$ operation.
\end{example}

\subsection{Constraint Solving}

 Let $Eq$ be a set of equality constraints, $Ineq$ be a set of subtyping constraints, and $Def$ a set of type definitions. Function $solve(Eq,Ineq,Def)$ solves the constraints, outputing a new set of type definitions. Note that the rewriting rules in the following definitions of the solver algorithm are assumed to be ordered. \newline

\noindent
$solve(Eq,Ineq,Def) = FinalDef$, where
\begin{itemize}
    \item $solve\_eq(Eq) = S$
    \item $solve\_ineq(S(Ineq),S(Def)) = FinalDef$
\end{itemize}

\begin{definition}
A set of equality constraints is in {\em solved form} if all constraints are of the form $\alpha = \tau$, where $\alpha$ is a type variable and $\tau$ is a type term.
\end{definition}

\subsubsection{Solving Equality Constraints}

Let $Eq$ be a set of equality constraints. Then $solve\_eq$ is a rewriting algorithm that outputs a set of equality constraints in solved form, which can be interpreted as a substitution, that will later be applied to $Ineq$ and $Def$.

We now present the rewrite rules for solving the type equality constraints. A {\em configuration} is either the term {\em fail} (representing failure) or a set of equality constraints {\em Eq}. The rewriting algorithm consists of the transformation rules on configurations listed below. 
\\
\\
\noindent
$solve\_eq(Eq) =$
\begin{enumerate}
    \item $\{t = t\} \cup Rest \to Rest$
    
    \item $\{\alpha = t\} \cup Rest \to \{\alpha = t\} \cup Rest[\alpha \mapsto t]$, if $\alpha$, which is a type variable, occurs somewhere in $Rest$
    
    \item $\{t = \alpha\} \cup Rest \to \{\alpha = t\} \cup Rest$, where $\alpha$ is a type var and $t$ is not a type variable nor a type symbol
    
    \item $\{f(t_1,\dots,t_n) = f(s_1,\dots,s_n)\} \cup Rest \to \{t_1 = s_1, \dots, t_n = s_n\} \cup Rest$
    
    \item otherwise $\to$ fail
\end{enumerate}

Note that an occur check would be required in step 2, but from the definition of our constraint generation algorithm, it never happens that $\alpha$ occurs in $t$.

\begin{example}
Following example \ref{mainexample}, applying $solve\_eq$ to the set of equality constraints is as follows step-by-step:
$$\{\alpha = [~], \beta = [\delta~|~\epsilon], \epsilon = \upsilon\} \rightarrow_2 \{\alpha = [~], \beta = [\delta~|~\upsilon]\}$$
And the final set only contains constraints in solved form, that can be seen as a substitution to be applied to $Ineq$ and $Def$.
\end{example}

\subsubsection{Solving Subtyping Constraints}

Let $Ineq$ be a set of inequality constraints and $Def$ be a set of type definitions. $solve\_ineq$ is a rewriting algorithm that outputs a new set of type definitions.

A {\em configuration} now is either the term {\em fail} (representing failure) or a pair of the form $(Ineq,Def)$, where $Ineq$ and $Def$ are as described above. The rewriting algorithm consists on the transformation rules on configurations listed below. In this algorithm we assume a store of pairs of types that already have been compared by the algorithm.
\newline
\newline
\noindent
$solve\_ineq(Ineq,Def) =$
\begin{enumerate}
    \item $(\{t \leq t\} \cup Rest,Def) \to (Rest,Def)$
    
    \item $(\{f(t_1,\dots,t_n) \leq f(s_1,\dots,s_n)\} \cup Rest, Def) \to (\{t_1 \leq s_1, \dots t_n \leq s_n\} \cup Rest,Def)$
    
    \item $(\{\alpha \leq t_1, \dots, \alpha \leq t_n\} \cup Rest, Def) \to (\{\alpha \leq t\} \cup Rest, S(Def))$,\\
    where $\alpha$ is a type variable, $n \geq 2$, and $intersect(t_1,\dots,t_n) = (t,S)$
    
    \item $(\{\alpha \leq t\} \cup Rest,Def) \to (Rest,Def [\alpha = t])$,\\
    where $\alpha$ is a type variable
    
    \item $(\{t_1 + \dots + t_n \leq t\} \cup Rest, Def) \to (\{t_1 \leq t, \dots, t_n \leq t\} \cup Rest, Def)$
    
    \item $(\{\sigma \leq t\} \cup Rest,Def) \to (Rest,Def)$,\\
    if $\sigma$ and $t$ have already been compared
    
    \item $(\{\sigma \leq t\} \cup Rest,Def) \to (\{Rhs_{\sigma} \leq t\} \cup Rest, Def)$,\\
    where $\sigma$ is a type symbol, and $\sigma = Rhs_{\sigma} \in Def$. Also add $(\sigma,t)$ to the store of pairs of types that have been compared.
    
    \item $(\{t_1 \leq \alpha, \dots t_n \leq \alpha\} \cup Rest,Def) \to (Rest,Def [\alpha = t_1 + \dots + t_n])$
    
    \item $(\{t \leq t_1 + \dots + t_n\}\cup Rest,Def) \to (\{t \leq t_i\}\cup Rest,Def)$,\\
    where $t_i$ is one of the summands
    
    \item $(\{t \leq \sigma\}\cup Rest,Def) \to (Rest,Def)$,\\
    if $\sigma$ and $t$ have already been compared.
    
    \item $(\{t \leq \sigma\}\cup Rest,Def) \to (\{t \leq Rhs_{\sigma}\}\cup Rest,Def)$,\\
    where $\sigma$ is a type symbol, and $\sigma = Rhs_{\sigma} \in Def$. Also add $(\sigma,t)$ to the store of pairs of types that have been compared.
    
    \item $(\emptyset,Def) \to Def$
    
    \item otherwise $\to$ fail
\end{enumerate}

Type intersection,  $intersect(t_1,\dots,t_n)$, is calculated using a version of Zobel's intersection algorithm previously presented in \cite{DBLP:conf/iclp/Zobel87}.

\begin{example}
Following example \ref{mainexample}, applying $solve\_ineq$ to the set of subtyping constraints is as follows step-by-step:\\
$$(\{\upsilon \leq \sigma_1\}, \{\sigma_1 = [~] + [\delta~|~\upsilon]\}) \rightarrow_4 (\emptyset, \{\sigma_1 = [~] + [\delta~|~\sigma_1]\})$$
\end{example}
{\bf Arithmetic} To extend the type inference algorithm to include built-in arithmetic predicates ($\leq/2$, $</2$, $\geq/2$, $>/2$, and $is/2$) we add the following rule to the constraint generation algorithm:
\begin{itemize}
    \item $generate(p(t_1,t_2)) = (bool,\{X_1 : \sigma_1, \dots , X_n:\sigma_n\},\emptyset,\{\alpha_1 \leq num, \dots , \alpha_n \leq num\}, \{\sigma_1 = \alpha_1, \dots, \sigma_n = \alpha_n\})$,\\
    where $X_1, \dots, X_n$ are the variables occurring in $t_1$ and $t_2$; $\sigma_i$ and $\alpha_i$ are fresh; and $num = int + float$.
\end{itemize}

\subsection{Decidability and Soundness}

The next theorem shows that both the equality constraint and subtyping constraint solvers terminate at every input set of constraints.
\begin{theorem}[Termination]
$solve\_eq$ and $solve\_ineq$ always terminate. When $solve\_eq$ terminates, it either fails or the output is in solved form. When $solve\_ineq$ terminates, it either fails or the output is a pair with zero subtyping constraints.
\end{theorem}
Proofs of this theorem follows a usual termination proof approach, where we show that a carefully chosen metric decreases at every step.

We do not have a formal proof of soundness yet, but we believe that proofs of soundness of type inference algorithms for functional programming can be adapted to the logic programming case; in particular, we need the following extra definitions and lemmas:

\begin{itemize}
    \item a type system which defines the notion of {\em well-typed} program: for this we use the type system defined in \cite{BarbosaFloridoCosta19};
    \item a notion of constraint satisfaction;
    \item a lemma for the soundness of constraint solving which states that the solver terminates in a set of constraints in solved form which satisfies the initial set of constraints.
\end{itemize}

Informally, the soundness theorem should state that 
if one applies the substitution corresponding to the solved form of the generated constraints to the type obtained by the constraint generation function, we get a well-typed program.

\section{Related Work} \label{relatedwork}

Types have been used before in Prolog systems: relevant works on type systems and type inference in logic programming include types used in the logic programming systems CIAO Prolog \cite{10.1007/BFb0026840,Vaucheret:2002:MPY:647171.718317}, SWI and Yap \cite{DBLP:conf/iclp/SchrijversCWD08}. CIAO uses types as approximations of the success set, while we use types as filters to the program semantics. There is an option where the programmer gives the types for the programs in the form of assertions, which is recommended in \cite{PCPH08}. The well-typings given in \cite{DBLP:conf/lopstr/SchrijversBG08}, also have the property that they never fail, in the sense that every program has a typing, which is not the case in our algorithm, which will fail for some predicates. The previous system of Yap only type checked predicate clauses with respect to programmer-supplied type signatures. 
Here we define a new type inference algorithm for pure Prolog, which is able to infer data types.

In several other previous works types approximated the success set of a predicate \cite{DBLP:conf/iclp/Zobel87,DBLP:books/mit/pfenning92/DartZ92,DBLP:books/mit/pfenning92/YardeniFS92,DBLP:conf/iclp/BruynoogheJ88}. This sometimes led to overly broad types, because the way logic programs are written can be very general and accept more than what was initially intended. These approaches were different from ours in the sense that in our work types can filter the success set of a predicate, whenever the programmer chooses to do so, using the closure operation.

A different approach relied on ideas coming from functional programming languages \cite{DBLP:journals/ai/MycroftO84,DBLP:conf/slp/LakshmanR91,DBLP:books/daglib/0095081,DBLP:conf/iclp/SchrijversCWD08}. Other examples of the influence of functional languages on types for logic programming are the type systems used in several functional logic programming languages \cite{DBLP:conf/birthday/Hanus13,DBLP:journals/jlp/SomogyiHC96}.
Along this line of research, a rather influential type system for logic programs was Mycroft and O'Keefe type system \cite{DBLP:journals/ai/MycroftO84}, which was later reconstructed by Lakshman and Reddy \cite{DBLP:conf/slp/LakshmanR91}. This system had types declared for the constants, function symbols and predicate symbols used in a program. Key differences from our work are: 1) in the Mycroft-O'Keefe type system, each clause of a predicate must have the same type. We lift this limitation extending the type language with sums of types, where the type of a predicate is the sum of the types of its clauses; 2) although we may use type declarations, they are optional and we can use a closure operation to infer datatype declarations from untyped programs.

Set constraints have also been used by many authors to infer types for logic programming languages \cite{DBLP:books/mit/pfenning92/HeintzeJ92,DBLP:conf/iclp/GallagherW94,Devienne97,Charatonik98,Drabent2000,Drabent02}. Although these approaches differ from ours since they follow the line of conservative approximations to the success set, we were inspired from several techniques from this area to define our type constraint solvers.

\section{Conclusions and Future Work}

In this paper, we present a flexible type inference algorithm for Prolog. Inferred types are standard semantic approximations by default, but the user may tune the algorithm, quite easily, to automatically infer types which correspond to the usual algebraic data types used in the program. Moreover, the algorithm may also be tuned to use predefined (optional) data type declarations to improve the output types. 
Although our  algorithm was tested and it was able to locate type errors in several existing Prolog programs, 
the next obvious step of this work is to prove its soundness with respect to a type system which defines the notion of {\em well-typed} program.
\bibliographystyle{alpha}
\bibliography{bibliography}
\newpage
\appendix

\section{Decidability}

\setcounter{theorem}{0}
\setcounter{lemma}{0}

\begin{lemma}[Equality Constraints]\label{lem1}
$solve\_eq$ always terminates and when it does, either it fails or the output is in normal form.
\end{lemma}
\begin{proof}
We will define the following metrics for $solve\_eq$:
\begin{itemize}
    \item NVRS: number of variables on the right-hand side of equalities, that occur somewhere in another equality.
    \item NSE: number of symbols in the equalities.
    \item NENNF: number of equalities not in normal form.
\end{itemize}

We will prove termination by showing that NENNF reduces to zero. Termination of $solve\_eq$ is proven by a measure function that maps the constraint set to a tuple (NVRS,NSE,NENNF). The following table shows that each step decreases the tuple w.r.t. the lexicographical order of the tuple.
\\
\\
\begin{tabular}{c c c c}
   & NVRS & NSE & NENNF \\
1. &  $\geq$ & $>$ & \\
2. & $\geq$ & $>$ & \\
3. & $=$ & $=$ & $>$ \\
4. & $=$ & $>$ & \\
5. &  0  &  0  & 0 \\
\end{tabular}
\\
\end{proof}
\begin{lemma}[Subtyping]\label{lem2}
$solve\_ineq$ always terminates.
\end{lemma}

\begin{proof}
We will define the following metrics for $solve\_ineq$:
\begin{itemize}
    \item NPC: number of possible comparisons between types that have not been made yet.
    \item NVRSI: number of variables on the right-hand side of inequalities.
    \item NSI: number of symbols in the inequalities.
    \item NI: number of inequalities.
\end{itemize}

We will prove termination by showing that NI reduces to zero. Termination of $solve\_ineq$ is proven by a measure function that maps the constraint set to a tuple (NPC,NVRSI,NSI,NI). The following table shows that each step decreases the tuple w.r.t. the lexicographical order of the tuple.
\\
\\
\begin{tabular}{c c c c c}
   & NPC  & NVRSI  & NSI  & NI \\
1. & $=$  & $\geq$ & $>$  &    \\
2. & $=$  & $=$    & $>$  &    \\
3. & $=$  & $\geq$ &$\geq$& $>$\\
4. & $=$  & $\geq$ & $>$  &    \\
5. & $=$  & $=$    & $>$  &    \\
6. & $=$  & $\geq$ & $>$  &    \\
7. & $>$  &        &      &    \\
8. & $=$  & $\geq$ & $>$  &    \\
9. & $=$  & $\geq$ & $>$  &    \\
10.& $=$  & $=$    & $>$  &    \\
11.& $>$  &        &      &    \\
12.& 0    &  0     & 0    & 0  \\
12.& $=$  &  0     & 0    & 0  \\
\end{tabular}
\end{proof}

\begin{theorem}[Termination]
$solve\_eq$ and $solve\_ineq$ always terminate. When $solve\_eq$ terminates, it either fails or the output is in solved form. When $solve\_ineq$ terminates, it either fails or the output is a pair with zero subtyping constraints.
\end{theorem}

\begin{proof}
By lemma \ref{lem1} and lemma \ref{lem2}, we prove this theorem.
\end{proof}

\end{document}